\newcommand{\shortversion}[1]{}
\newcommand{\longversion}[1]{#1}

\shortversion{
  \documentclass{llncs} 
}
\longversion{
  \documentclass[10pt,usletter]{article}
  \usepackage[totalwidth=450pt,totalheight=640pt]{geometry}
  \date{}
  \usepackage{amsthm}
}

\usepackage{amsfonts}
\usepackage{graphicx,tikz}
\usetikzlibrary{arrows}
\usepackage{xcolor}
\usepackage{boxedminipage}

\usepackage{amssymb,amsmath}

\newtheorem{LEM}{Lemma} 
\newtheorem{THE}{Theorem}

\newtheorem{PRO}{Proposition} 
 
\newtheorem{DEF}{Definition}

\newcommand{\defparproblem}[4]{
  \vspace{1mm}
\noindent\fbox{
  \begin{minipage}{0.96\textwidth}
  \begin{tabular*}{\textwidth}{@{\extracolsep{\fill}}lr} #1  & {\bf{Parameter:}} #3 \\ \end{tabular*}
  {\bf{Input:}} #2  \\
  {\bf{Question:}} #4
  \end{minipage}
  }
  \vspace{1mm}
}

\newcommand{\defproblem}[3]{
  \vspace{1mm}
\noindent\fbox{
  \begin{minipage}{0.96\textwidth}
  \begin{tabular*}{\textwidth}{@{\extracolsep{\fill}}lr} #1  &  \\ \end{tabular*}
  {\bf{Input:}} #2  \\
  {\bf{Question:}} #3
  \end{minipage}
  }
  \vspace{1mm}
}

\shortversion{
  \title{Parameterized Algorithms for Modular-Width}
  \author{Jakub Gajarsk\'{y}\inst{1}\thanks{Research funded by the Czech Science Foundation under grant
P202/11/0196} \and Michael Lampis\inst{2} \and Sebastian Ordyniak\inst{1}\thanks{Research funded by 
      Employment of Newly Graduated Doctors of Science for
      Scientific Excellence (CZ.1.07/2.3.00/30.0009).}
  }
  \institute{
    Faculty of Informatics, Masaryk University, 
    Brno, Czech Republic, \\
    \email{\{gajarsky,ordyniak\}@fi.muni.cz}. 
    \and
    KTH Royal Institute of Technology,
    Stockholm, Sweden\\
    \email{mlampis@kth.se}
  }
}
\longversion{
  \title{Parameterized Algorithms for Modular-Width}
  \author{Jakub Gajarsk\'{y}$^1$\thanks{Research funded by the Czech Science Foundation under grant
P202/11/0196} \quad Michael Lampis${}^2$ \quad Sebastian Ordyniak$^{1}$\thanks{Research funded by 
      Employment of Newly Graduated Doctors of Science for
      Scientific Excellence (CZ.1.07/2.3.00/30.0009).}\\[8pt]
    \small ${}^1$ Faculty of Informatics, Masaryk University, Brno, Czech Republic\\[-3pt]
    \small \texttt{\{gajarsky,ordyniak\}@fi.muni.cz}\\
    \small ${}^2$ KTH Royal Institute of Technology, Stockholm, Sweden\\[-3pt]
    \small \texttt{mlampis@kth.se}\\
  }
}

\usepackage{boxedminipage}
\usepackage{boxedminipage}
\usepackage{xcolor}
\usepackage{framed}

\newcommand{\FPT}{\text{\normalfont FPT}}

\newcommand{\W}[1][xxxx]{\text{\normalfont W[#1]}}
\newcommand{\N}{\mathbb{N}}

\newcommand{\vc}{\textup{vc}}
\newcommand{\nd}{\textup{nd}}
\newcommand{\tc}{\textup{tc}}
\newcommand{\sd}{\textup{sd}}
\newcommand{\mw}{\textup{mw}}
\newcommand{\tw}{\textup{tw}}
\newcommand{\cw}{\textup{cw}}

\newcommand{\SB}{\{\,}
\newcommand{\SM}{\;{:}\;}
\newcommand{\SE}{\,\}}
\newcommand{\Card}[1]{|#1|}

\newcommand{\ham}{\textup{ham}}

\newcommand{\ILP}{\textup{ILP}}
\newcommand{\INTEGER}{\mathcal{Z}}

\def\TM#1#2{{\cal T\!M}_{#2}(#1)}

\begin{document}

\maketitle

\begin{abstract}

It is known that a number of natural graph problems which are FPT parameterized
by treewidth become W-hard when parameterized by clique-width. It is therefore
desirable to find a different structural graph parameter which is as general as
possible, covers dense graphs but does not incur such a heavy algorithmic
penalty.

The main contribution of this paper is to consider a parameter called
modular-width, defined using the well-known notion of modular decompositions.
Using a combination of ILPs and dynamic programming we manage to design FPT
algorithms for Coloring and Partitioning into paths (and hence Hamiltonian path and Hamiltonian cycle), which are W-hard for both clique-width and its
recently introduced restriction, shrub-depth. We thus argue that modular-width
occupies a sweet spot as a graph parameter, generalizing several simpler
notions on dense graphs but still evading the ``price of generality'' paid by
clique-width.

\end{abstract}

\section{Introduction}

The topic of this paper is the exploration of the algorithmic properties of
some structural graph parameters. This area is typically dominated by an effort
to achieve two competing goals: generality and algorithmic tractability. A good
example of this tension is the contrast between treewidth and clique-width. 

A large wealth of problems are known to be FPT when parameterized by treewidth
\cite{BodlaenderK08,Bodlaender06,Bodlaender00}. One drawback of treewidth,
however, is that this parameterization excludes a large number of interesting
instances, since, in particular, graphs of small treewidth are necessarily
sparse. The notion of clique-width (and its cousins rank-width \cite{Oum05} and
boolean-width \cite{Bui-XuanTV11}) tries to ameliorate this problem by covering
a significantly larger family of graphs, including many dense graphs.  As it
turns out though, the price one has to pay for this added generality is
significant. Several natural problems which are known to be fixed-parameter
tractable for treewidth become W-hard when parameterized by these measures
\cite{FominGLS10,FominGLS10b,FominGLS09}. 

It thus becomes an interesting problem to explore the trade-offs offered by
these and other graph parameters. More specifically, one may ask: is there a
natural graph parameter which covers dense graphs but still allows FPT
algorithms for the problems lost to clique-width? This is the main question
motivating this paper. We first attempt to use the recently introduced notion
of shrub-depth for this role \cite{GanianHNOMR12}. Shrub-depth is a restriction
of clique-width which shows some hope, since it has been used to obtain
improved algorithmic meta-theorems.  Unfortunately, as we will establish, 
the hardness constructions for \textsc{Coloring} and \textsc{Hamiltonian path}
used in~\cite{FominGLS10} go through with small
modifications for this restricted parameter as well.

The main contribution of this paper is then the consideration of a parameter
called modular-width which, we argue, nicely fills this niche. One way to
define modular-width is by using the standard concept of modular decompositions
(see e.g. \cite{McConnellS99}), as the maximum degree of the optimal modular
decomposition tree.  As a consequence, a graph's modular-width can be computed
in polynomial time.  Note that the concept of modular-width was already briefly
considered in \cite{CourcelleMR00}, but was then abandoned in that paper in
favor of the more general clique-width. To the best of our knowledge,
modular-width has not been considered as a parameter again, even though modular
decompositions have found a large number of algorithmic applications, including
in parameterized complexity (see \cite{habib2010survey} for a general survey
and \cite{protti2009applying,DomGHNT06,BessyPP10} for example applications in
parameterized complexity). 

We give here the first evidence indicating that modular-width is a structural
parameter worthy of further study. On the algorithmic side, modular-width
offers a significant advantage compared to clique-width, a fact we demonstrate
by giving FPT algorithms for several variants of \textsc{Hamiltonicity} and
\textsc{Chromatic number}, all problems W-hard for clique-width. At the same
time, we show that modular-width significantly generalizes several simpler
parameters, such as neighborhood diversity \cite{Lampis12} and twin-cover
\cite{Ganian11}, which also allow FPT algorithms for these problems.

Our main algorithmic tool is a form of dynamic programming on the modular
decomposition of the input graph. Unlike dynamic programming on the more
standard tree decompositions, the main obstacle here is in combining the DP
tables of the children of a node to compute the table for the node itself. This
is in general a hard problem, but we show that it can sometimes be made
tractable if every node of the decomposition has small degree, hence the
parameterization by modular-width. 

Even if the modular decomposition has small degree, combining the DP tables is
still not necessarily a trivial problem. A second idea we rely on (in the case
of \textsc{Hamiltonicity}) is to use an Integer Linear Program, whose number of
variables is bounded by the number of modules we are trying to combine. It is
our hope that this technique, which seems quite general, will be applicable to
other problems as well.

\shortversion{\paragraph{Full Version} Proofs of statements
  marked with ($\star$) are shortened or omitted due to space restrictions.
  Detailed proofs can be found in the full
  version, available at \url{arxiv.org/abs/1308.2858}.}

\section{Preliminaries}

We use standard notation from graph theory as can be found in,
e.g.,~\cite{diestel00}. Let $G$ be a graph. We denote the vertex set of $G$
by $V(G)$ and the edge set of $G$ by $E(G)$. 
Let $X \subseteq V(G)$ be a set of vertices of $G$. The \emph{subgraph of
  $G$ induced by $X$}, denoted $G[X]$, is the graph with vertex set $X$
and edges $E(G) \cap [X]^2$. By $G \setminus X$ we denote
the subgraph of $G$ induced by $V(G) \setminus X$. Similarly, for $Y
\subseteq E(G)$ we define $G \setminus Y$ to be the subgraph of $G$
obtained by deleting all edges in $Y$ from $G$. For a
graph $G$ and a vertex $v \in V(G)$, we denote by $N_G(v)$ and $N_G[v]$
the open and closed neighborhood of $v$ in $G$,
respectively. 

\subsection{Considered Problems}\label{ssec:cons-prob}

We consider the following problems on graphs.
Let $G$ be a graph. A \emph{coloring} of $G$ is a function 
$\lambda : V(G) \rightarrow \N$ such that for every edge $\{u,v\} \in
E(G)$ it holds that  $\lambda(u)\neq \lambda(v)$. We denote by
$\lambda(G)$ the set of colors used by
the coloring $\lambda$, i.e., $\lambda(G)=\SB \lambda(v) \SM v \in
V(G) \SE$, and by
$\Lambda(G)$ the set of all colorings of $G$ that use at most $|V(G)|$
colors. The \emph{chromatic number} of $G$, denoted by $\chi(G)$,
is the smallest number $c$ such that $G$ has a coloring $\lambda$ with
$|\lambda(G)| \leq c$. 

\defproblem{\textsc{Graph Coloring}}{A graph $G$.}{Compute $\chi(G)$.}

Let $G$ be a graph. A partition of $G$ into paths is a set of disjoint 
paths of $G$ whose union contains every vertex of $G$. We denote by
$\ham(G)$ the least integer $p$ such that $G$ has a partition into 
$p$ paths.

\defproblem{\textsc{Partitioning Into Paths}}{A graph $G$.}{Compute
  $\ham(G)$.}

\defproblem{\textsc{Hamiltonian Path}}{A graph $G$.}{Does $G$ have a
  Hamiltonian Path?}

\defproblem{\textsc{Hamiltonian Cycle}}{A graph $G$.}{Does $G$ have a
  Hamiltonian Cycle?}

\subsection{Parameterized Complexity}

Here we introduce the relevant concepts of parameterized complexity theory.
For more details, we refer to text books on the topic~\cite{DowneyFellows99,FlumGrohe06,Niedermeier06}.
An instance of a parameterized problem is a pair $(I,k)$ where $I$ is
the main part of the instance, and $k$ is the parameter.  A
parameterized problem is \emph{fixed-parameter tractable} if instances
$(I,k)$ can be solved in time $f(k)\Card{I}^c$, where $f$ is a
computable function of $k$, and $c$ is a constant.  $\FPT$ denotes the
class of all fixed-parameter tractable problems.  Hardness for
parameterized complexity classes is based on \emph{fpt-reductions}.  A
parameterized problem $L$ is fpt-reducible to another parameterized
problem~$L'$ if there is a mapping $R$ from instances of $L$ to
instances of $L'$ such that (i) $(I,k) \in L$ if and only if $(I',k')
= R(I,k) \in L'$, (ii) $k' \leq g(k)$ for a computable function $g$,
and (iii) $R$ can be computed in time $O(f(k)\Card{I}^c)$ for a
computable function $f$ and a constant~$c$.  
% If $L=L'$, and $f(k)$ is
% a polynomial, and there is a function $h$ such that $\Card{I'}\leq
% h(k)$, then the fpt-reduction is called a \emph{kernelization} and $h$
% is called the kernel size. If $h$ is a polynomial then we say that the
% problem $L$ admits a \emph{polynomial kernel}.  
Central to the
completeness theory of parameterized complexity is the hierarchy $\FPT
\subseteq \W[1] \subseteq \W[2] \subseteq \dots $.  Each
intractability class $\W[t]$ contains all parameterized problems that
can be reduced to a certain parameterized satisfiability problem under
fpt-reductions.

\subsection{Treewidth}

The treewidth of a graph is defined using the following notion of a
tree decomposition (see, e.g., \cite{Bodlaender93}).
A \emph{tree decomposition} of an (undirected) graph 
$G=(V,E)$ is a pair $(T,\chi)$ where $T$ is a tree and $\chi$ is a labeling
function that assigns each tree node $t$ a set $\chi(t)$ of vertices of
the graph $G$  such that the
following conditions hold: 
(1)~Every vertex of $G$ occurs in $\chi(t)$ for some tree node~$t$,
(2)~For every edge $\{u,v\}$ of $G$ there is a tree node
$t$ such that $u,v\in \chi(t)$, and (3)~For every vertex $v$ of $G$,
the tree nodes $t$ with $v\in \chi(t)$ form a connected subtree
of~$T$.
The
\emph{width} of a tree decomposition $(T,\chi)$ is the size of a largest bag
$\chi(t)$ minus~$1$ among all nodes~$t$ of~$T$.  A tree decomposition of
smallest width is \emph{optimal}.  The \emph{treewidth} of a graph $G$ is the
width of an optimal tree decomposition of~$G$.

\subsection{Shrub-depth}
The recently introduced notion of {\em shrub-depth} \cite{GanianHNOMR12} is
the ``low-depth'' variant of clique-width, similar to the role that tree-depth
plays with respect to treewidth.

%Clique width of a graph $G$ is defined as the smallest number of labels $k=cw(G)$
%such that $G$ can be constructed using operations to
%create a new vertex with label $i$, take the disjoint union of
%two labelled graphs, add all edges between vertices of label $i$ and
%label $j$, and relabel all vertices with label $i$ to have label $j$.
%Similarly as tree-depth is related to tree-width,
%there exists a very new notion of {\em shrub-depth} \cite{GanianHNOMR12}
%which is (in a sense) related to clique-width,
%and which we explain next.

\begin{DEF}
\label{def:tree-model}
We say that a graph $G$ has a {\em tree-model of $m$ colors and depth $d\geq1$}
if there exists a rooted tree $T$ (of height $d$) such that 
\begin{enumerate}
 \item the set of leaves of $T$ is exactly $V(G)$,
\item the length of each root-to-leaf path in $T$ is exactly~$d$,
\item each leaf of $T$ is assigned one of $m$ colors
(this is not a graph coloring, though),
\item\label{it:tree-model-edge}
and the existence of a $G$-edge between $u,v\in V(G)$ depends solely
on the colors of $u,v$ and the distance between $u,v$ in $T$.
\end{enumerate}
The class of all graphs having a tree-model of $m$ colors and depth $d$ 
is denoted by $\TM dm$.
\end{DEF}

\begin{DEF}
\label{def:shrub-depth}
A class of graphs $\mathcal{G}$ has {\em shrub-depth} $d$
if there exists $m$ such that $\mathcal{G}\subseteq\TM dm$,
while for all natural $m$ it is $\mathcal{G}\not\subseteq\TM{d-1}m$.
\end{DEF}

Note that Definition~\ref{def:shrub-depth} is asymptotic as it makes sense 
only for infinite graph classes.
Particularly, classes of shrub-depth $1$ are known as the graphs of bounded
{\em neighborhood diversity} in~\cite{Lampis12},
i.e., those graph classes on which the twin relation on pairs of vertices
(for a pair to share the same set of neighbors besides this pair)
has a finite index.

\longversion{
  \subsection{Properties of Shrupth-depth}

  In this section we will show some technical results that will later help us
  for our hardness proofs.
  For two graphs $G$ and $H$, let $G \bowtie H$ denote the set of graphs such that they
  consist of a copy of $G$, a copy of $H$ (disjoint from $G$) and arbitrary edges between
  vertices of $G$ and $H$. 

  \begin{PRO}
    \label{pro:sd1}
    Let $G$ be a graph with tree-model $\textup{TM}_G$ of height $d$ which uses $m$ colors and let $H$
    be a graph. Then there 
    exists, for every $G' \in G \bowtie H$, a tree-model of height $d$ which uses at most
    $m \cdot 2^{|V(H)|} + |V(H)|$ colors.
  \end{PRO}
  \begin{proof}
    For simplicity, we refer to copies of $G$ and $H$ in $G'$ by $G$ and $H$. We fix
    arbitrary ordering of $V(H)$.

    Since each vertex of $G$ is adjacent to some subset of $V(H)$ in $G'$, there are at most $2^{V(H)}$
    ``types'' of vertices in $G$ according to their neighborhoods in $H$. We extend
    the colors used in the leaves of $\textup{TM}_G$ by $|V(H)|$ bits and for each leaf $l$ of
    $\textup{TM}_G$, we set these bits as follows: we set the $i$-th bit to $1$ iff the vertex 
    $v$ of $G$ which corresponds to $l$ in $\textup{TM}_G$ is adjacent to the $i$-th vertex of $H$
    in $G'$. We denote this 
    new tree-model by $\textup{TM}^+_G$.

    Next, we construct a tree-model $\textup{TM}_H$ of height one for $H$ -- it consists of
    a root vertex and $V(|H|)$ leaves, each of different color (we use new colors,
    different from those of $\textup{TM}_G^+)$.

    Now we construct a tree-model $\textup{TM}_{G'}$ of $G'$ -- we connect the root of $\textup{TM}_H$ 
    to the root of $\textup{TM}_G^+$ by a path of length $d-1$ (to ensure that the leaf-to-root distance is the same in the whole tree-model; if the height of $\textup{TM}_G^+$ is one,
    we identify the root of $\textup{TM}_H$ with the root of $\textup{TM}_G^+$.)  It is not hard to see
    that $\textup{TM}_{G'}$ is a tree-model of $G'$:
    \begin{itemize}
    \item the edges in the ``$H$''-part of $G'$ depend only on the colors of the subtree $\textup{TM}_H$ of $\textup{TM}_{G'}$,
    \item the edges in the ``$G$''-part of $G'$ depend only on the colors of the subtree $\textup{TM}^+_G$ of
      $\textup{TM}_{G'}$ and the original distances in $\textup{TM}_G$ (simply by looking at the colors of leaves before adding
      new bits, i.e. ignoring additional bits),
    \item the edges between the vertices of ``$H$''- and ``$G$''-parts depend only on the colors used
      in the ``$H$''-part and the newly added bits in the ``$G$''-part of $G'$ (all are at distance $2d$
      in $\textup{TM}_{G'}$),
    \item none of these three edge-dependencies affect the other two.
    \end{itemize} 
    \shortversion{\qed}\end{proof}

  \begin{PRO}
    \label{pro:sd3}
    Let $G$ be a graph, $S$ a subset of its vertices, and $\textup{TM}_G$ its tree-model of height
    $d_G$ and $m_G$ colors. Let $H$ be a graph, $R$ a subset of its vertices, and $\textup{TM}_H$ its tree-model of height $d_H$ and $m_H$ colors. Assume that $d_H \le d_G$. Then the graph $G'$ obtained
    by taking a copy of $G$, a copy of $H$, and making each vertex from $S$ adjacent to each
    vertex of $R$ has a tree-model $\textup{TM}_{G'}$ of height $d$ and at most $2m_G + 2m_H$ colors.
  \end{PRO}

  \begin{proof}
    We assume that the sets of colors in $\textup{TM}_G$ and $\textup{TM}_H$ are disjoint. The proof uses the
    same idea as the proof of Proposition~\ref{pro:sd1}. First, we extend the colors in both
    $\textup{TM}_H$ and $\textup{TM}_G$ by one additional bit. For a leaf of $\textup{TM}_G$, we set this bit to 1 iff
    the corresponding vertex of $G$ is in $S$. Similarly, for a leaf of $\textup{TM}_H$, we set this
    bit to 1 iff the corresponding vertex of $H$ is in $R$. We denote the tree-models obtained in
    this way by $\textup{TM}^+_G$ and $\textup{TM}^+_H$.

    Now we construct $\textup{TM}_{G'}$ by connecting the root of $\textup{TM}_H$ to the root of $\textup{TM}_G$ by a path
    of length $d_G - d_H$ (or identify these two roots if $d_G = d_H$). 
    The distances between
    the pairs of leaves in both ``$G$''- and ``$H$''-parts of $\textup{TM}_{G'}$ are the same as they 
    were in $\textup{TM}_G$ and $\textup{TM}_H$; their original colors can be ``recovered'' by simply ignoring
    the newly added bit.
    The existence of an edge between a vertex from $G$ and a vertex from $H$ in $G'$ is determined
    by $\textup{TM}_{G'}$ by the value of the additional bit and the fact that their colors come from different
    tree-models (the sets of colors were disjoint); the distance between them is always $2d$.
    Thus, $\textup{TM}_{G'}$ is indeed a tree-model of $G'$. Since the numbers $d_G$ and $d_H$ were
    doubled by adding a new bit to each color and since $\textup{TM}_{G'}$ does not use any other color,
    the Proposition is now proved.

    \shortversion{\qed}\end{proof}

  \begin{PRO}
    \label{pro:sd2}
    Let $G$ be a graph with a tree-model $\textup{TM}_G$ of height $d$ and $m$ colors and 
    let $R$ be a subset of $V(G)$. Let $H$ be a 
    graph and $S$ be a subset of $V(H)$. Let $G'$ be the graph obtained from $G$ by
    creating $|R|$ copies of $H$ indexed
    by vertices of $R$ and making each $v \in R$ adjacent to the vertices $S_v$ of $H_v$.
    Then $G'$ has a tree-model $\textup{TM}_{G'}$ of height $d+1$ and $m + |V(H)|$ colors.
  \end{PRO}
  \begin{proof}
    We extend the set of colors used in $\textup{TM}_G$ by $|V(G)|$ new colors.
    The construction of $\textup{TM}_{G'}$ from $\textup{TM}_G$ is simple -- for each $v \in R$, replace its
    corresponding leaf $l$ in $\textup{TM}_G$ by a tree of height one which has:
    \begin{itemize}
    \item leaf $l'$ with the same color as $l$ 
    \item $|V(H)|$ leaves, each with different color from the newly created colors
    \end{itemize}
    For each $v \not\in R$, subdivide once the edge to which its corresponding leaf was 
    adjacent in $\textup{TM}_G$ (to have the same leaf-to-root distance in the whole $\textup{TM}_{G'}$).

    It is easy to see that $\textup{TM}_{G'}$ is indeed a tree-model of $G'$ -- for each $v \in R$
    the leaves corresponding to $v$ and $V(H_v)$ in $\textup{TM}_{G'}$ are at distance two from each
    other and the adjacency depends only on the colors.
    The color of each leaf corresponding to $v \in V(G)$ remained the same, and
    the distance between each pair of such leaves was increased by 2.
    \shortversion{\qed}\end{proof}

  \begin{LEM}
    \label{lem:tw_sd} Let $\mathcal{G}$ be a class of graphs which admit
    a tree-decomposition of width $w$ and height $d$. Then $\mathcal{G}$
    has shrub-depth $d+1$.
  \end{LEM}
  \begin{proof} 
    
    Let $G \in \mathcal{G}$.
    We proceed by induction on the height $d$ of a tree-decomposition $(T,\chi)$
    of $G$ of width at most $w$.

    If the
    height is $0$, then the tree-decomposition consists of $1$ bag which
    contains at most $w$ vertices. A tree-model of height $1$ for the graph
    induced by this bag can be easily constructed -- create a tree of
    height $1$, with $1$ leaf for every vertex in the bag and
    assign each leaf a different color. For graphs of treewidth $w$, these
    tree-models always use at most $w$ colors.

    If the height of $T$ is $d$, we take a root bag $B_r$ (by root bag we mean
    any bag whose distance to any other bag is at most $d$) of the tree-decomposition
    and create a tree-model $\textup{TM}_r$ of height one for the graph induced
    by the bag; this requires at most $w$ colors. After deleting
    this bag from the
    decomposition, we are left with a collection of tree-decompositions of
    height $d-1$. We delete all vertices of $B_r$ from all bags of these 
    decompositions.
    By the induction hypothesis, the graphs induced by these
    decompositions have tree-models of height $d$ with at most $m$ colors
    (for some natural $m$, which depends only on $d$ and $w$).
    Let us denote these tree-models by $\textup{TM}_i$.
    
    Now we construct a tree-model $\textup{TM}^-$ of height $d+1$ by introducing
    a root vertex $r$ and connecting the root of each $\textup{TM}_i$ to $r$. This
    is a tree-model of $G \setminus B_r$, which uses $m$ colors. Since
    $G \in G[B_r] \bowtie G[G \setminus B_r]$, we can use Proposition~\ref{pro:sd1}
    to conclude that there exists a tree-model $\textup{TM}_G$ of $G$ of height $d+1$
    which uses at most $m \cdot 2^{w} + w$ colors. Since neither $m$ or $w$ depend
    on particular choice of $G$, this concludes the proof.
    \shortversion{\qed}\end{proof}
}

\subsection{Modular-Width}\label{subsec:modular-width}

For our algorithms we consider graphs that can be obtained from an
algebraic expression that uses the following operations:
\begin{itemize}
\item[] (O1) create an isolated vertex;
\item[] (O2) the disjoint union of $2$ graphs, i.e., 
  the \emph{disjoint union} of $2$ graph $G_1$ and $G_2$,
  denoted by $G_1 \oplus G_2$, is the graph with vertex
  set $V(G_1) \cup V(G_2)$ and edge set $E(G_1) \cup E(G_2)$;
\item[](O3) the complete join of $2$ graphs, i.e., the \emph{complete
    join} of $2$ graphs $G_1$ and $G_2$, denoted by $G_1 \otimes G_2$,
  is the graph with vertex set $V(G_1) \cup V(G_2)$ and edge set 
  $E(G_1) \cup E(G_2) \cup \SB \{v,w\} \SM v \in V(G_1) \textup{ and }w
  \in V(G_2) \SE$;
\item[](O4) the substitution operation with respect to some graph $G$ with
  vertices $v_1,\dots,v_n$, i.e., for graphs $G_1,\dots,G_n$ the
  \emph{substitution} of the vertices of $G$ by the graphs
  $G_1,\dots,G_n$, denoted by $G(G_1,\dots,G_n)$, is the graph with
  vertex set $\bigcup_{1\leq i \leq n}V(G_i)$ and edge set $\bigcup_{1
    \leq i \leq n}E(G_i) \cup \SB \{u,v\} \SM u \in V(G_i) \textup{
    and } v \in V(G_j) \textup{ and }i\neq j\SE$. Hence,
  $G(G_1,\dotsc,G_n)$ is obtained from $G$ by substituting every
  vertex $v_i \in V(G)$ with the graph $G_i$ and adding all edges
  between the vertices of a graph $G_i$ and the vertices of a graph
  $G_j$ whenever $\{v_i,v_j\} \in E(G)$. 

\end{itemize}
Let $A$ be an algebraic expression that uses only the operations
(O1)--(O4). We define the \emph{width} of $A$ as the maximum number of
operands used by any occurrence of the operation (O4) in $A$.
It is well-known that the \emph{modular-width} of a graph $G$, denoted
$\mw(G)$, is the least integer $m$ such that $G$ can be obtained from
such an algebraic expression of width at most $m$. Furthermore, an
algebraic expression of width $\mw(G)$ can be constructed in linear time
~\cite{TedderCorneilHabibPaul08}.

\subsection{Integer Linear Programming}

For our algorithms, we use the well-known result that \textsc{Integer
  Linear Programming} is fixed-parameter tractable parameterized by
the number of variables. 

\defparproblem{\textsc{Integer Linear Programming Feasibility}}
{A matrix $A \in \INTEGER^{m \times p}$ and a vector $b \in
  \INTEGER^{m}$.}
{p}
{Is there a vector $x \in \INTEGER^{p}$ such that $Ax \leq b$?}

\begin{PRO}[\cite{FellowsLokshtanocMisraRosamondSaurabh08}]\label{pro:ILPF}
  \textsc{Integer Linear Programming Feasibility} is fixed-parameter
  tractable and can be solved in time $O(p^{2.5p+o(p)}\cdot L)$ 
  where $L$ is the number of bits in the input.
\end{PRO}

% \defparproblem{\textsc{Integer Linear Programming Optimization}}
% {A matrix $A \in \INTEGER^{m \times p}$ and $2$ vectors $b \in
%   \INTEGER^{m}$ and $c \in \INTEGER^{p}$.}
% {p}
% {A vector $x \in \INTEGER^{p}$ that maximizes (the scalar product) $c
%   \cdot x$ such that $Ax \leq b$.}

% \begin{PRO}[\cite{FEL_LP}]\label{pro:ILPO}
%   \textsc{Integer Linear Programming Optimization} is fixed-parameter
%   tractable and can be solved in time $O(p^{2.5p+o(p)}\cdot L \cdot
%   \log (M \cdot N))$ where $L$ is the number of bits in the input, $N$
%   is the maximum of the absolute values any variable can take, and $M$
%   is an upper bound on the absolute value of the optimum value of the
%   objective function.
% \end{PRO}

\section{Hardness for Problems on Shrub-depth}

In this section we give evidence that the recently introduced
parameter shrub-depth is not restrictive enough to obtain fixed-parameter
algorithms for problems that are \W[1]-hard on graphs of bounded
cliquewidth. 
%This is surprising to us as shrub-depth captures the
%non-elementary behavior of $\msol_1$-definable
%problems on bounded cliquewidth graphs~\cite{GanianHNOMR12}. 
In particular, we show that
\textsc{Graph Coloring} and \textsc{Hamiltonian Path} are \W[1]-hard
parameterized by the number of colors (used in a tree-model of the
input graph) on classes of graphs of shrub-depth $5$. 
Note that restricting the shrub-depth means restricting the height of
the tree-model that can be employed and for every restriction on the
height of the tree-model the number of colors needed to model the
graph gives a different parameter. In particular, if we restrict the
shrub-depth to $1$ the number of colors of the tree-model corresponds
to the neighborhood diversity of a graph. This implies
that \textsc{Graph Coloring} and \textsc{Hamiltonian Path} become
fixed-parameter tractable when parameterized by the number of colors
(used in a tree-model of the input graph) on classes of graphs of
shrub-depth $1$~\cite{Lampis12}. It is an interesting open question what is
the least possible shrub-depth that allows for fixed-parameter
algorithms for the problems \textsc{Graph Coloring} and
\textsc{Hamiltonian Path}. 

\shortversion{
  \begin{THE}[$\star$]\label{the:coloring-hard}
    \textsc{Graph Coloring} parameterized by the number of colors (used in a
    tree-model of the input graph) is \W[1]-hard on classes of graphs of
    shrub-depth $5$.
  \end{THE}
}
\longversion{
  \begin{THE}\label{the:coloring-hard}
    \textsc{Graph Coloring} parameterized by the number of colors (used in a
    tree-model of the input graph) is \W[1]-hard on classes of graphs of
    shrub-depth $5$.
  \end{THE}
  We will prove the above theorem after showing the following lemma.
  \begin{LEM}
    \label{lem:col} \textsc{Equitable Coloring} is \W[1]-hard
    parameterized by the variant of treewidth, where the height of the
    tree (of the tree-decomposition) is at most $3$.
  \end{LEM}
  \begin{proof}
    Our proof uses the same reduction as the proof of~\cite[Theorem
    5.]{COL_HARD}. We only show that the tree-decompositions of
    the graphs constructed there have height at most $3$.

    The graph $G'$ constructed in~\cite{COL_HARD} consists of a disjoint
    union of trees of height $2$. Hence, the  tree-decomposition of $G'$ has
    height at most $3$. The next step in their
    construction is to add an appropriate number of isolated vertices to
    $G'$. This step does not increase the height of the tree-decomposition
    any further.
    
    The final step of their reduction consists of adding a clique and
    connecting some vertices of the clique to some vertices of the graph
    constructed above. Hence, a tree-decomposition of height
    $3$ for the whole graph can be obtained by adding all vertices of
    the clique to every bag of the already constructed
    tree-decomposition.
    \shortversion{\qed}\end{proof}

  \begin{proof}[Proof of Theorem~\ref{the:coloring-hard}]
    Our proof uses the same reduction from \textsc{Equitable Coloring}
    on graphs of bounded treewidth as the proof of~\cite[Theorem
    3.2]{FominGLS09}. We only show that the graphs constructed there
    have tree-models of height at most $5$ that use at most $f(w)$ colors,
    where the input graph $G$ has a tree-decomposition of width $w$ and
    height $3$.

    For the input graph $G$, we
    consider its tree-model $\textup{TM}_G$ of height $4$ and with at most
    $m$ colors (for some $m$ which does not depend on $G$) whose existence
    is guaranteed by Lemma~\ref{lem:tw_sd} and Lemma~\ref{lem:col}.  We
    replace each leaf of $\textup{TM}_G$ by a tree of height one which has
    the following leaves:
    \begin{itemize}
    \item vertex $u$ which was replaced from the original tree-model, with
      the same color as $u$ had in $\textup{TM}_G$
    \item $r$ vertices colored by $S_1$ to $S_r$
    \item $r$ vertices colored by $P_1$ to $P_r$
    \end{itemize} We create a tree-model for the clique $P_M$ of size $r$, where
    each vertex has a different color $p_i$, $1\le i \le r$. Finally, we
    create tree-models of height 1 for cliques $C_i$ ($1\le i \le r$) of
    size $n\frac{r-1}{r}$, each of them fully colored by color $c_i$.
    
    For each $i \in \{1,\ldots r\}$, the union of all vertices colored by $S_i$ ($P_i$)
    corresponds
    to the set $S_i$ ($P_i$) in~\cite[Theorem3.2]{FominGLS09}, respectively. 
    Similarly graphs $P^M, C_1, \ldots C_r$ correspond to their counterparts
    of the same name in~\cite[Theorem3.2]{FominGLS09}.
    
    We connect the roots of tree-models
    of $P^M, C_1, \ldots C_r$ to the root of $\textup{TM}_G$ (by a path
    of appropriate length, so that the leaf-to-root distance is the same
    in the whole tree) to obtain $\textup{TM}_G'$.

    The height of $\textup{TM}_G'$ is $5$, because the height of
    $\textup{TM}_G$ was $4$ and during the construction the
    height increased by one (by replacing leaves by
    trees of height one).

    Now we describe the edges and their dependencies in $\textup{TM}'_G$:
    \begin{itemize}
    \item all edges from the original graph remain the same, $\textup{TM}'$ clearly
      encodes this information, since the only thing that changed is that
      the distance between every pair $u,v$ of leaves from the original
      tree-model increased by $2$;
    \item we make each $u$ adjacent to each vertex with label $P_i$ (for
      all $i \in \{1, \ldots , r \}$, in the whole tree-model), i.e. for all
      distances between leaves, all the vertices colored by colors of the original
      tree-model are adjacent to all the vertices colored by $P_i$;
    \item we make each $u$ adjacent to ``its'' vertices colored by $S_i$,
      (for all $i \in \{1, \ldots , r \})$ i.e. for distance $2$, all
      vertices which have colors from the original tree-model are adjacent
      to all the vertices colored by $S_i$, $i \in \{1, \ldots , r \}$;
    \item we make each vertex colored by $S_i$ adjacent to the whole clique
      $P_M$ except for a vertex colored by $p_i$, i.e. for distance $10$,
      each vertex colored by $S_i$ is adjacent to each vertex colored by
      $p_j$, except when $i=j$;
    \item for distance 2, vertices colored $P_i$ and $S_j$ are not
      adjacent if $i=j$, and are adjacent if $i \neq j$. For distance higher
      than $2$, vertices colored by $P_i$ and $S_j$ are adjacent, whether
      $i=j$ or not;
    \item we make, for each $i$, every vertex of $C_i$ adjacent to all
      vertices of colored by $S_i$, i.e. for distance 10, vertices colored
      by $c_i$ and $S_i$ are adjacent, for $i \in \{1, \ldots , r \}$
    \item we make, every vertex of $C_i$ adjacent to all vertices colored
      by $P_j$, $i \neq j$, i.e.  for distance 10, vertices colored by
      $c_i$ and $P_j$, for all $i,j \in \{1, \ldots , r \}$ are adjacent,
      except when $i=j$.
    \end{itemize}
    The reduction from ~\cite[Theorem3.2]{FominGLS09} is now completed, and
    $\textup{TM}_G'$ represents the resulting graph of this reduction.
    We used at most $m$ colors of the tree-model of original graph and $4r$ new colors
    $S_i, P_i, c_i, p_i$. The number of colors in $\textup{TM}_G'$ therefore bounded
    by $m + 4r$, which concludes the proof.
    \shortversion{\qed}\end{proof}
}

\shortversion{
  \begin{THE}[$\star$]\label{ham_hard}
    \textsc{Hamiltonian Path} parameterized by the number of colors
    (used in a tree-model of the input graph) is $W[1]$-hard on class of graphs
    of shrub-depth $5$.
  \end{THE}
}
\longversion{
  \begin{THE}\label{ham_hard}
    \textsc{Hamiltonian Path} parameterized by the number of colors
    (used in a tree-model of the input graph) is $W[1]$-hard on class of graphs
    of shrub-depth $5$.
  \end{THE}
  We observe that the reduction showing $W[1]$-hardness of 
  \textsc{Capacitated Dominating Set} on bounded treewidth given
  in~\cite{CAP_DOM_SET}
  actually shows that this problem is $W[1]$-hard with respect to
  the variant of treewidth, where the height of the
  tree (of the tree-decomposition) is at most $2$.

  We then proceed by showing that the reduction given
  in~\cite{FominGLS10} from \textsc{Capacitated Dominating Set}
  parameterized by treewidth to
  \textsc{Hamiltonian Path} parameterized by cliquewidth produces graphs
  of shrub-depth at most $5$.

  \begin{PRO}\label{cap_dom_set} 
    \textsc{Capacitated Dominating Set} is $W[1]$-hard
    when parametrized by the variant of treewidth, where the height of the
    tree (of the tree-decomposition) is at most $2$.
  \end{PRO}
  \begin{proof} 
    Our proof uses the same reduction as the proof of~\cite[Theorem
    1.]{CAP_DOM_SET}. We only show that the tree-decompositions of
    the graphs constructed there have height at most $2$.
    Indeed, in the graph $H$
    constructed in~\cite[Section 3.1.]{CAP_DOM_SET}, delete the sets $S_i$
    and $S_{i,j}$ and all the vertices $x_i$, $y_i$, $z_i$, $p_{i,j}$,
    $q_{i,j}$. After this, we are left with a forest of height $1$, which
    has a tree-decomposition of height $2$. After adding all deleted
    vertices to each bag, we obtain a tree-decomposition of the whole
    graph of height $2$ and of width $O(k^4)$.
    \shortversion{\qed}\end{proof}

  \begin{proof}[Proof of Theorem~\ref{ham_hard}]
    We use the reduction from~\cite[Theorem 5.1]{FominGLS10}. We show that when the
    instances of \textsc{Capacitated Dominating Set} have a tree-decomposition of
    height $2$, then the graphs we obtain by their reduction have
    shrub-depth at most $5$.

    First we argue that, when we begin with the graph $G$ which has a
    tree-decomposition of height $2$, the graph $G'(c')$ has a		
    tree-decomposition of height $3$. The graph $G'(c')$ is the graph obtained
    after the first step of the reduction with an assumption that the capacity of each
    vertex is exactly one. This graph is basically obtained by replacing
    each vertex of $G$ by a gadget and replacing each edge of $G$ by
    two gadgets. Since these gadgets have constant size, one can just
    replace each occurrence of vertex $v$ in the tree-decomposition of $G$ by a
    gadget, and ``hang'' the bag containing the pair of gadgets corresponding to an edge
    $u,v$ in $G$ to any bag which contained these two vertices in the original
    tree-decomposition. This increases the height of the tree-decomposition by at most one.

    By Lemma~\ref{lem:tw_sd}, the resulting graph has a tree-model $\textup{TM}_1$ of
    height $4$ and at most $m$ colors, for some natural $m$ (which depends only on
    the height of the newly created tree-decomposition (3) and its width $w'$, which is 
    bounded by a constant multiple of $w$). 
    
    Graph $G'(c)$ is obtained from $G'(c')$ by adding
    vertices, each of which is a twin of some vertex from $G'(c')$. 
    Even though this can in general affect the number of colors in the tree-model
    of a graph, in this particular instance it is not the case -- each vertex 
    to which we add twin vertices has two neighbors and is the only vertex which
    has this neighborhood. We can therefore construct the tree-model $\textup{TM}_2$ of
    $G'(c)$ by simply
    adding siblings (with the same color) to the corresponding leafs of $\textup{TM}_1$.

    The second part of the reduction consists of introducing a graph which consists
    of 2 new vertices connected by an appropriate number of paths of length
    two. Let $Y$ denote the set of middle vertices of these paths. 
    All vertices from $Y$ are connected to some subset of $V(G'(c))$.
    It is easy to construct a tree-model of height 1 with 2 colors for this graph
    (root, one color for vertices of $Y$ and one color for remaining two vertices)
    and therefore by Proposition~\ref{pro:sd2} there exists a tree-model $\textup{TM}_3$
    of the resulting graph of height $4$ and with $2m + 4$ colors. 
    
    Next, to some subset of vertices of the graph induced by $\textup{TM}_3$, a copy of a gadget of
    fixed size $c$ is connected (in the same fashion for all vertices). The tree-model
    of this gadget of height 1 with $c$ colors is easy to construct and therefore 
    by Proposition~\ref{pro:sd3} the tree-model $\textup{TM}_4$ of the graph obtained in
    this way has height $5$ and at most $2m + 4 + c$ colors.
    
    After this, two specific vertices $x$ and $y$ of each of these newly added gadgets are
    connected to the whole $Y$. By the construction from the proof of Proposition~\ref{pro:sd3}
    each vertex from $Y$ now has one of two colors in  $\textup{TM}_4$ and no vertex outside
    $Y$ has any of these colors. By the construction from the proof of
    Proposition~\ref{pro:sd2} the copies of vertices $x$ and $y$ have the same colors
    in the whole $\textup{TM}_4$, and no other vertex has this color. It follows that
    these newly added edges depend only on the colors of the leaves of the tree-model.

    The final step of reduction consists of adding an independent set,
    which is connected to some subset of $G'$. An independent set
    has tree-model of height $1$ and $1$ color and therefore by Proposition~\ref{pro:sd3}
    the resulting graph has tree-model $\textup{TM}_5$ of height 5 and at most $2(2m + 4 + c) + 2$ colors.
    
    Since neither $m$ nor $c$ depend on the choice of $G$, the theorem is now proved.
    \shortversion{\qed}\end{proof}
}

%Theorems~\ref{the:coloring-hard} and~\ref{ham_hard} suggest that 

\section{Modular-Width and Other Parameters}

In this section we study the relationships of modular-width,
shrub-depth and other important width parameters. Of particular
importance is the observation that modular-width generalizes the
recently introduced parameters neighborhood diversity~\cite{Lampis12}
and twin-cover~\cite{Ganian11}. Both of these parameters have been
introduced to obtain FPT algorithms on dense graphs for
problems that are hard for
clique-width. Figure~\ref{fig:width-parameters} summarizes these
relationships. Most of these relationships are well-known or have
recently been shown
in~\cite{Lampis12,Ganian11,GanianHNOMR12,CourcelleMR00}. Consequently,
we only show the relationships whose proofs cannot been found anywhere
else.
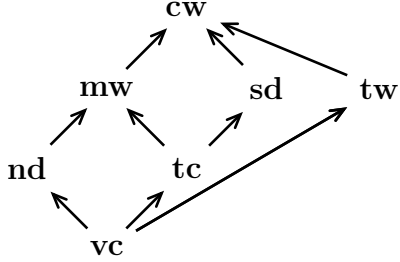
\begin{figure}
  \begin{center}
    \begin{tikzpicture}[node distance=1.5cm,>=angle 45]
      \tikzstyle{every node}=[]
      \tikzstyle{every edge}=[draw,line width=1pt,->]
      \large \bf
      \draw
      node (vc) {\vc}
      node[above left of=vc] (nd) {\nd}
      node[above right of=vc] (tc) {\tc}
      node[above right of=nd] (mw) {\mw}
      node[above right of=tc] (sd) {\sd}
      node[right of=sd] (tw) {\tw}
      node[above right of=mw] (cw) {\cw}

      (vc) edge (nd)
      (vc) edge (tc)
      (vc) edge (tw)
      (nd) edge (mw)
      (tc) edge (mw)
      (tc) edge (sd)
      (vc) edge (tw)
      (tw) edge (cw)
      (sd) edge (cw)
      (mw) edge (cw)
      ;
    \end{tikzpicture}
  \end{center}
  \caption{Relationship between vertex cover (vc), neighborhood
    diversity (nd), twin-cover (tw), modular-width (mw), shrub-depth
    (sd), treewidth (tw), and clique-width (cw). Arrows indicate
    generalizations, e.g., modular-width generalizes both neighborhood
  diversity and twin-cover.}
  \label{fig:width-parameters}
\end{figure}
\begin{THE}\label{the:rel-mw-nd-tc}
  Let $G$ be a graph. Then $\mw(G) \leq \nd(G)$ and $\mw(G) \leq 2^{\tc(G)}+\tc(G)$.
  Furthermore, both inequalities are strict, i.e., there are graphs
  with bounded modular-width and unbounded neighborhood diversity (or
  unbounded twin-cover number).
\end{THE}
\begin{proof}
  Let $G$ be a graph. Using the definition of neighborhood diversity
  from~\cite{Lampis12} it follows that $G$ has a partition
  $\{V_1,\dots,V_{\nd(G)}\}$ of its vertex set such that for every $1
  \leq i \leq \nd(G)$ it holds that the graph $G[V_i]$ is
  either a clique or an independent set and for every $1 \leq i < j \leq
  \nd(G)$, either all vertices in $V_i$ are adjacent to all vertices
  of $V_j$ or $G$ contains no edges between vertices in $V_i$ and
  vertices in $V_j$. Let $G'$ be the graph with vertex set
  $v_1,\dots,v_{\nd(G)}$ and an edge between $v_i$ and $v_j$ if and
  only if the graph $G$ contains all edges between vertices in $V_i$
  and vertices in $V_j$. Then
  $G=G'(G[V_1],\dots,G[V_{\nd(G)}])$. Furthermore, because for every
  $1 \leq i \leq \nd(G)$, $V_i$ is either a clique or an independent
  set, we can obtain the graph $G[V_i]$ from an algebraic expression $A_i$
  that uses only the operations (O1)--(O3). Substituting the
  algebraic expressions $A_i$ for every $G[V_i]$ into
  $G'(G[V_1],\dots,G[V_{\nd(G)}])$ gives us the desired algebraic
  expression for $G$ of width $\nd(G)$.

  Let $G$ be a graph. Using the definition of twin-cover
  from~\cite{Ganian11} it follows that there is a set $C$ of at most
  $\tc(G)$ vertices of $G$ such that every component $C'$ of $G
  \setminus C$ is a clique and every vertex in $C'$ is connected to
  the same vertices in $C$. Let $C_1,\dots,C_l$ be sets of components of $G
  \setminus C$ such that $2$ components of $G \setminus C$ are
  contained in the same set $C_i$ if and only if their vertices have
  the same neighbors in $C$. Because there are at most $2^{|C|}$ 
  possible such neighborhoods, we obtain that $l \leq 2^{|C|}$. Let $G'$ be the graph with vertices
  $C \cup \{c_1,\dots,c_l\}$ and edges $E(G[C]) \cup \SB \{c_i,v_i\}
  \SM v_i \in N_G(C_i) \SE$. Then
  $G=G'(v_1,\dots,v_{|C|},G[C_1],\dots,G[C_l])$. Furthermore, because
  the graphs $G[C_i]$ are disjoint unions of cliques, we can obtain
  each of these graphs from an algebraic expression $A_i$ that uses
  only the operations (O1)--(O3). Substituting the
  algebraic expressions $A_i$ for every $G[C_i]$ into
  $G'(v_1,\dots,v_{|C|},G[C_1],\dots,G[C_l])$ gives us the desired algebraic
  expression for $G$ of width $2^{|C|}+|C|\leq 2^{\tc(G)}+\tc(G)$.

  To see that both inequalities are strict we refer the reader
  to~\cite[Example 5.4 a)]{GanianHNOMR12}. The example exhibits a
  family of co-graphs, i.e., graphs of modular-width $0$, with unbounded
  neighborhood diversity and unbounded twin-cover number.
\shortversion{\qed}\end{proof}
The following theorem shows that modular-width and shrub-depth are
orthogonal to each other.
\shortversion{
  \begin{THE}[$\star$]\label{the:rel-mw-sd}
    There are classes of graphs with unbounded modular-width and bounded
    shrub-depth and vice versa.
  \end{THE}
}
\longversion{
  \begin{THE}\label{the:rel-mw-sd}
    There are classes of graphs with unbounded modular-width and bounded
    shrub-depth and vice versa.
  \end{THE}
  \begin{proof}
    Let $S_n$ be the graph obtained from a star (with $n$ leaves) after
    subdividing every edge exactly once and let $\mathcal{S}$ be the
    class of all these graphs. Because $S_n$ is a prime graph, its
    modular-width equals its number of vertices, i.e.,
    $2n+1$. Consequently, $\mathcal{S}$ has unbounded modular-width. We
    next show that $\mathcal{S}$ has shrub-depth at most $2$ by showing
    that $S_n$ has a tree-model of height $2$ that uses at most $3$
    colors. The tree-model consists of a root $r$ that has $n+1$
    children, such that $n$ of these children (which correspond to the
    $n$ edges around the center of the star) have $2$ children themselves
    colored by color $1$ and $2$, respectively, and the remaining child
    has only $1$ child (which corresponds to the center of the star) 
    which is colored by color $2$. Then $S_n$ can be defined in this
    tree-model.
    
    To show the ``vice versa'' part of the theorem, we refer the reader
    to \cite[Example 5.4 a)]{GanianHNOMR12}. The example gives a family
    of co-graphs, i.e., graphs with modular-width $0$, that has
    unbounded shrub-depth.
    \shortversion{\qed}\end{proof}
}
The next theorem shows that also shrub-depth generalizes
neighborhood diversity and twin-cover.
\begin{THE}\label{the:rel-tc-nd-sd}
  Let $\mathcal{G}$ be a class of graphs. If $\mathcal{G}$ has bounded
  neighborhood diversity or bounded twin-cover number then
  $\mathcal{G}$ has shrub-depth at most $2$.  
  Furthermore, there are classes of
  graphs that have unbounded neighborhood diversity and twin-cover
  number but shrub-depth $2$.
\end{THE}
\begin{proof}
  Suppose that $\mathcal{G}$ has bounded neighborhood diversity, i.e.,
  there is some natural number $c$ such that $\nd(G) \leq c$ for every
  $G \in \mathcal{G}$. We show that every graph $G \in
  \mathcal{G}$ has a tree-model of height at most $1$ that uses at
  most $\nd(G)\leq c$ colors. Using the definition of neighborhood diversity
  from~\cite{Lampis12} it follows that $G$ has a partition
  $\{V_1,\dots,V_{\nd(G)}\}$ of its vertex set such that for every $1
  \leq i \leq \nd(G)$ it holds that the graph $G[V_i]$ is
  either a clique or an independent set and for every $1 \leq i < j \leq
  \nd(G)$, either all vertices in $V_i$ are adjacent to all vertices
  of $V_j$ or $G$ contains no edges between vertices in $V_i$ and
  vertices in $V_j$. Then the tree-model for $G$ consists of a root
  $r$ and $1$ leave for every vertex $v$ in $G$ with color $i$ if  $v
  \in V_i$.

  Now, suppose that $\mathcal{G}$ has bounded twin-cover, i.e.,
  there is some natural number $c$ such that $\tc(G) \leq c$ for every
  $G \in \mathcal{G}$. We show that every graph $G \in
  \mathcal{G}$ has a tree-model of height at most $2$ that uses at
  most $2^{\tc(G)}+tc(G)\leq 2^c+c$ colors. Using the definition of twin-cover
  from~\cite{Ganian11} it follows that there is a set
  $W=\{w_1,\dots,w_{\tc(G)}\}$ 
  of $\tc(G)$ vertices of $G$ such that every component $C'$ of $G
  \setminus W$ is a clique and every vertex in $C'$ is connected to
  the same vertices in $W$. Let
  $C_1^1,\dots,C_{p_1}^1,\dotsc,C_1^l,\dots,C_{p_l}^l$ be all the components of
  $G \setminus W$ such that $2$ components $C_{i_1}^{j_1}$ and
  $C_{i_2}^{j_2}$ have the same neighborhood in $W$ if and only if
  $j_1=j_2$.
  Because there are at most $2^{|W|}$ 
  possible such neighborhoods, we obtain that $l \leq 2^{|W|}$.
  We construct a tree-model of $G$ as follows. We start with the root
  node $r$, which has $1$ child, say $C_i^j$, for every component of
  $G \setminus W$, and $1$ child (which is also a leaf of the
  tree-model), say $w_i$, with color $l+i$ for every $1 \leq i \leq
  |W|$. Finally, every node $C_i^j$
  has $|V(C_i^j)|$ children (which are also leaves of the tree-model)
  with color $j$. This finishes the construction of the tree-model for
  $G$.

  To see that there are classes of
  graphs that have unbounded neighborhood diversity and twin-cover
  number but shrub-depth $2$, consider the class $\mathcal{S}$
  from the proof of Theorem~\ref{the:rel-mw-sd}. As shown in
  Theorem~\ref{the:rel-mw-sd} $\mathcal{S}$ has unbounded
  modular-width but shrub-depth $2$. It now follows from
  Theorem~\ref{the:rel-mw-nd-tc} that $\mathcal{S}$ has also unbounded
  neighborhood diversity and unbounded twin-cover number.
\shortversion{\qed}\end{proof}

\section{Algorithms on Modular-width}

In this section we show that \textsc{Partitioning Into Paths},
\textsc{Hamiltonian Path}, \textsc{Hamiltonian Cycle}, and
\textsc{Coloring} are fixed-parameter tractable parameterized by the
modular-width of the input graph. Our algorithms use a bottom-up dynamic
programming approach along the parse-tree of an algebraic expression
as defined in Section~\ref{subsec:modular-width}. That is for every
node of such a parse-tree we compute a solution (or a record
representing a solution) given solutions (or records) for the
children of the node in the parse-tree. 
The running time of our algorithms is then the number of nodes in the
parse-tree times the maximum time spend at any node of the parse-tree.
Because the number of nodes of a parse-tree is linear
in the number of vertices of the created graph, it suffices to bound the maximum time
spend at any node of the parse-tree. Furthermore, because the
operations (O1)--(O3) can be replaced by one operation of type (O4) that
uses at most $2$ operands, we only need to bound the time spend to
compute a record for the graph obtained by operation (O4).
To avoid cumbersome run-time bounds we use the notation $O^+$ to
suppress poly-logarithmic factors, i.e., we write $O^+(f)$ when we have 
$O(f\log^df)$ for some constant $d$.

\subsection{Coloring}

This section is devoted to a proof of the following
theorem. Recall the definition of \textsc{Graph Coloring} and related
notions from Section~\ref{ssec:cons-prob}.
\begin{THE}
  \textsc{Graph Coloring} parameterized by the modular-width of the
  input graph is fixed-parameter tractable.
\end{THE}
As outlined above we only need to bound the time spend to compute a
record for a node of type (O4) of the parse-tree. In the case
of \textsc{Graph Coloring} a record is simply the chromatic
number of the graph. Hence, we will have shown the theorem after
showing the following lemma.
\begin{LEM}\label{lem:coloring-of-product}
  Let $G$ be a graph with vertices $v_1,\dots,v_n$, $G_1,\dotsc,G_n$
  be graphs, and $H:=G(G_1,\dotsc,G_n)$. 
  Then $\chi(H)$ can be computed from $\chi(G_1),\dotsc,\chi(G_n)$ in
  time % $O(n^{n+1}(\log(\max_{1 \leq i \leq n}\chi(G_i))))$.
  $O^+(2^nn^2\max_{1 \leq i \leq n}\chi(G_i))$.
\end{LEM}
We will prove the lemma by reducing the coloring problem to the
following problem.

\defproblem{\textsc{Max Weighted Partition}}{An $n$-element set $N$
  and functions $f_1,\dotsc,f_k$ from the subsets of $N$ to integers
  from the range $[-M,M]$.}{A $k$-partition $(S_1,\dotsc,S_k)$ of $N$
  that maximizes $f_1(S_1)+\dotsb+f_k(S_k)$.}
\begin{PRO}\label{pro:max-weighted-partition}(\emph{\cite[Theorem 4.]{BjorklundHK09}})
  \textsc{Max Weighted Partition} can be solved in time $O^+(2^nk^2M)$.
\end{PRO}
To simplify the reduction to \textsc{Max Weighted Partition}
we need the following Proposition and Lemma. 
\begin{PRO}[\cite{Lampis12}]\label{pro:coloring-neighborhood-diversity}
  Let $G$ be a graph with vertices $v_1,\dotsc,v_n$ and
  $s_1,\dotsc,s_n$ be natural numbers.
  Then $\chi(G(K_{s_1},\dotsc,K_{s_n}))=\min_{\lambda \in
    \Lambda(G)}(\sum_{c \in \lambda(G)}\max \SB s_i \SM v_i \in \lambda^{-1}(c)\SE)$.
\end{PRO}
\shortversion{
  \begin{LEM}[$\star$]\label{lem:coloring-product}
    Let $G$ be a graph with vertices $v_1,\dots,v_n$, $G_1,\dotsc,G_n$
    be graphs, $H_K:=G(K_{\chi(G_1)},\dotsc,K_{\chi(G_n)})$, and $H:=G(G_1,\dotsc,G_n)$. 
    Then $\chi(H_K)=\chi(H)$.
  \end{LEM}
}
\longversion{
  \begin{LEM}\label{lem:coloring-product}
    Let $G$ be a graph with vertices $v_1,\dots,v_n$, $G_1,\dotsc,G_n$
    be graphs, $H_K:=G(K_{\chi(G_1)},\dotsc,K_{\chi(G_n)})$, and $H:=G(G_1,\dotsc,G_n)$. 
    Then $\chi(H_K)=\chi(H)$.
  \end{LEM}
  \begin{proof}
    We will show that for every coloring $\lambda$ of $H$ there is a
    coloring $\lambda_K$ of $H'$ that uses no more colors, i.e.,
    $|\lambda_K(H_K)|\leq|\lambda(H)|$, and vice versa. 
    Let $\lambda$ be a coloring for $H$. Then for every $1 \leq i \leq
    n$ the number of colors used to color the copy of
    $G_i$ in $H$ is at least $\chi(G_i)$. Hence, we can use a subset of
    these colors to color the copy of $K_{\chi(G_i)}$ in $H_K$.
    
    For the reverse direction, let $\lambda_K$ be a coloring for
    $H_K$. Then for every $1 \leq i \leq
    n$ the number of colors used to color the copy of
    $K_{\chi(G_i)}$ in $H_K$ is $\chi(G_i)$. Hence, we can use the same
    colors to color the copy of $G_i$ in $H$.
    \shortversion{\qed}\end{proof}
}
We can now proceed with a proof of Lemma~\ref{lem:coloring-of-product}.
\begin{proof}[of Lemma~\ref{lem:coloring-of-product}]\sloppypar
  We reduce the coloring problem to the \textsc{Max Weighted
    Partition} problem as follows: We set $N:=V(G)$ and
  $f_1(S)=\dotsb=f_k(S)=-\max \SB \chi(G_i) \SM v_i \in S \SE$ for
  every subset $S$ of $N$. It
  follows from Proposition~\ref{pro:coloring-neighborhood-diversity} and
  Lemma~\ref{lem:coloring-product} that the maximum weight of a
  partition of this instance corresponds to the chromatic number
  $\chi(H)$. Hence, the lemma follows from
  Proposition~\ref{pro:max-weighted-partition}.

  % It follows from
  % Proposition~\ref{pro:coloring-neighborhood-diversity} and
  % Lemma~\ref{lem:coloring-product} that $\chi(H)$ can
  % be computed from the numbers $\chi(G_1),\dotsc,\chi(G_m)$ by
  % iterating over all possible colorings of the graph $G$ and for each
  % of these colorings $\lambda$ computing the number
  % $\sum_{c \in \lambda(G)}\max \SB \chi(G_i) \SM v_i \in \lambda^{-1}(c)\SE$.
  % Because there are at most $n^n$ possible colorings for $G$ and for
  % each of these colorings $\lambda$ we
  % can compute $\sum_{c \in \lambda(G)}\max \SB \chi(G_i) \SM v_i \in \lambda^{-1}(c)\SE$
  % in time $O(n(\log(\max_{1 \leq i \leq m}\chi(G_i))))$ the total time to compute
  % $\chi(H)$ from $\chi(G_1),\dots,\chi(G_n)$ is
  % $O(n^{n+1}(\log(\max_{1 \leq i \leq n}\chi(G_i))))$, as required.
\shortversion{\qed}\end{proof}

\subsection{Partitioning into Paths}

This section is devoted to a proof of the following theorem. Recall
the definition of \textsc{Partitioning Into Paths} and related
notions from Section~\ref{ssec:cons-prob}.
\begin{THE}
  \textsc{Partitioning Into Paths} (and hence also \textsc{Hamiltonian
    Path} and \textsc{Hamiltonian Cycle}) parameterized by the modular
  width of the input graph is fixed-parameter tractable.
\end{THE}
As outlined above we only need to bound the time spend to compute a
record for a node of type (O4) of the parse-tree. In the case
of \textsc{Partitioning Into Path} a record of a graph $G$ is the pair
$(\ham(G),|V(G)|)$. Hence, we will have shown the theorem after
showing the following lemma. From now on we will assume that
$G$ is a graph with vertices $v_1,\dots,v_n$, $G_1,\dotsc,G_n$
are graphs, $H=G(G_1,\dotsc,G_n)$, and $m=|E(G)|$. 
\begin{LEM}\label{lem:hamiltonian-of-product}\sloppypar
  Given the graph $G$ and the pairs 
  $(\ham(G_1),|V(G_1)|),\dotsc,(\ham(G_n),|V(G_n)|)$ the pair
  $(\ham(H),|V(H)|)$ can be computed in time 
  $O^+\left(\ham\left(H\right)\left((m+n)2^{n} + n\left(2(m+n)\right)^{5(m+n)+o(m+n)}\right)\right)$
  % $O^+\left(\ham\left(H\right)m\left(2^{n} + n\left(2m\right)^{5m+o(m)}\right)\right)$.
  % $O(\ham(H) \cdot(s + (2m)^{5m+o(m)} \cdot \log s))$ where
  % $s=m(2^{|V(G)|+1}+\log(\max_{1 \leq i \leq
  %   n}|V(G_i)|))$ and $m=|E(G)|+|V(G)|$.
\end{LEM}
The remainder of this section is devoted to a proof of this lemma.

For a graph $G$ and an integer $i$ we define the
graph $G \oplus i$ as the graph with vertex set $V(G) \cup \{1,\dotsc,i\}$
and edge set $E(G) \cup \SB \{v,j\} \SM v \in V(G) \textup{ and }
1\leq j \leq i\SE$, i.e., the graph $G \oplus i$ is obtained from $G$ by
adding $i$ vertices and connect them to every
vertex in $G$.
\shortversion{
  \begin{PRO}[$\star$]\label{pro:hp-hc}
    Let $G$ be a graph and \\$h(G) = \min \SB i \SM G \oplus i \textup{ has a
      Hamiltonian cycle} \SE$. Then $\ham(G)=h(G)$.
  \end{PRO}
}
\longversion{
  \begin{PRO}\label{pro:hp-hc}
    Let $G$ be a graph and \\$h(G) = \min \SB i \SM G \oplus i \textup{ has a
      Hamiltonian cycle} \SE$. Then $\ham(G)=h(G)$.
  \end{PRO}
  \begin{proof}
    We first show that $h(G) \leq \ham(G)$. Let $\{P_1, \dots,
    P_{\ham(G)}\}$ be a partition of $G$ into paths. Then $G \oplus \ham(G)$
    contains the Hamiltonian cycle
    $1$, $P_1$, $2$, $P_2, \dots,\ham(G)$, $P_{\ham(G)}$, $1$, as required.
    It remains to show that $\ham(G) \leq h(G)$. Let $C$ be a
    Hamiltonian cycle in $G\oplus h(G)$. Then $C \setminus \{1, \dots,h(G)\}$
    is a partition of $G$ into at most $h(G)$ disjoint path, as required.
    \shortversion{\qed}\end{proof}
}
A slightly less general version of the following lemma has already been
proven in~\cite{Lampis12}. 
\begin{LEM}\label{lem:hamiltonian-cycle-and-ILP}
  Let
  \textsc{Hamiltonian Cycle} be the \ILP{} with variables 
  $\SB e_{ij},e_{ji} \SM \{v_i,v_j\} \in E(G) \SE$ and constraints:
  \begin{tabbing}
    (1) $\sum_{j \in \SB l \SM v_l \in N_G(v_i) \SE}e_{ij}=\sum_{j \in
      \SB l \SM v_l \in N_G(v_i) \SE}e_{ji}$ \= \kill

    For every $1 \leq i \leq n$:\\

    (1) $\sum_{j \in \SB l \SM v_l \in N_G(v_i) \SE}e_{ij}=\sum_{j \in
      \SB l \SM v_l \in N_G(v_i) \SE}e_{ji}$ \> (``incoming = outgoing'')\\

    (2) $\sum_{j \in \SB l \SM v_l \in N_G(v_i) \SE}e_{ij} \leq
    |V(G_i)|$ \> (at most $|V(G_i)|$)\\

    (3) $\ham(G_i) \leq \sum_{j \in \SB l \SM v_l \in N_G(v_i)
      \SE}e_{ij}$  \> (at least $\ham(G_i)$)\\

    For every partition of $V(G)$ into vertex sets $A$ and $B$:\\

    (4) $\sum_{1 \leq i < j \leq n \SM \{v_i,v_j\} \in E(G) \land |e
      \cap A|=1}e_{ij}+e_{ji}\geq 1$ \> (``connectivity'')\\

    For every variable $e_{ij}$:\\
    
    (5) $e_{ij} \geq 0$.\\
  \end{tabbing}
  Then $H$ has a Hamiltonian cycle if and only if the ILP
  \textsc{Hamiltonian Cycle} is feasible.
  Furthermore, the size of the \ILP{} is at most
  $O^+(m2^{n})$ 
%  $O(|E(G)|(2^{|V(G)|}+\log(\max_{1 \leq i \leq m}|V(G_i)|)))$ 
  and it has $2m$ variables.
\end{LEM}
\begin{proof}
  The size bound on the \ILP{} \textsc{Hamiltonian Cycle} is obvious.
  Suppose that $H$ has a Hamiltonian cycle $C$. W.l.o.g. we can assume
  that $C$ is directed. For every $\{v_i,v_j\} \in E(G)$ we
  set $e_{ij}$ to be the number of arcs $(x,y)$ in $C$ such that $x
  \in V(G_i)$ and $y \in V(G_j)$ and similarly we set $e_{ji}$ to be
  the number of arcs $(x,y)$ in $C$ such that $x \in V(G_j)$ and $y
  \in V(G_i)$. Then, because $C$ is a Hamiltonian cycle of $H$, this
  assignment of $e_{ij}$ and $e_{ji}$
  satisfies the constrains (1)--(5), as required.  

  For the reverse direction, suppose that the \ILP{}
  \textsc{Hamiltonian Cycle} is feasible and let $\beta$ be an
  assignment of the variables $e_{ij}$ and $e_{ji}$ witnessing this.
  Let $G'$ be the directed multigraph obtained from $G$ by replacing
  every edge $\{v_i,v_j\}$ with $\beta(e_{ij})$ parallel arcs from $v_i$ to
  $v_j$ and $\beta(e_{ji})$ parallel arcs from $v_j$ to $v_i$. Because of the
  constrains (1), (4), and (5), it follows that $G'$ contains a directed
  eularian tour $T$, i.e., a closed directed walk that visits all the
  arcs of $G'$ exactly once. Clearly, when fixing any vertex of $G'$,
  the tour $T$ defines an ordering of the arcs of $G'$. Let $\pi$ be
  any such ordering of the arcs of $G'$.
  For every $1 \leq i \leq n$, let
  $\mathcal{P}_i=(P^i_1,\dots,P^i_{p_i})$
  be a partition of $G_i$ into $p_i$ disjoint paths, 
  where $p_i=\sum_{j \in \SB l \SM v_l \in N_G(v_i) \SE}e_{ij}$. Because of the
  constrains (2) and (3) we know that such a partition exists for every
  $1 \leq i \leq n$. For every arc $a=(v_i,v_j)$ in $T$ where $a$ is the $l$-th arc
  leaving $v_i$ in $T$ and $a$ is
  the $l'$-th arc entering $v_j$ in $T$ (according to the ordering
  $\pi$), we denote by $e(a)$ the edge of $H$ from the second endpoint of
  $P^i_l$ to the first endpoint of $P^j_{l'}$. Then the edges in $\SB
  e(a) \SM a \in T \SE$ together with the edges of all the path
  $P^1_1,\dots,P^n_{p_n}$ form a Hamiltonian cycle in $H$, as required.
 \shortversion{\qed}\end{proof}
\begin{LEM}\label{lem:hamiltonian-cycle-of-product}\sloppypar
  Given the graph $G$ and the pairs 
  $(\ham(G_1),|V(G_1)|),\dotsc,(\ham(G_n),|V(G_n)|)$ it can be decided
  whether the graph $H$ has a Hamiltonian cycle in time 
  $O^+(m2^{n} + (2m)^{5m+o(m)}n)$.
  % $O(s + (2|E(G)|)^{5|E(G)|+o(|E(G)|)} \cdot \log s)$ where
  % $s=|E(G)|(2^{|V(G)|}+\log(\max_{1 \leq i \leq n}|V(G_i)|))$.
 \end{LEM}
\begin{proof}\sloppypar
  To decide whether the graph $H$ has a Hamiltonian cycle we construct
  and solve the \ILP{} \textsc{Hamiltonian Cycle} from
  Lemma~\ref{lem:hamiltonian-cycle-and-ILP}. The running time of this
  algorithm is then the time it takes to construct the \ILP{}, i.e.,
  $O^+(m2^{n})$, plus the time needed to solve the \ILP{},
  i.e., $O^+((2m)^{5m+o(m)} \log (m2^{n})) \in O^+((2m)^{5m+o(m)} n)$
  using
  Proposition~\ref{pro:ILPF}. This concludes the
  proof of the Lemma.
\shortversion{\qed}\end{proof}
We are now ready to show Lemma~\ref{lem:hamiltonian-of-product}.
\begin{proof}[Proof of Lemma~\ref{lem:hamiltonian-of-product}]
  Clearly, $|V(H)|=\sum_{1 \leq i \leq n}|V(G_i)|$ so it remains to
  show how to compute $\ham(H)$. Because of
  Proposition~\ref{pro:hp-hc} $\ham(H)$ is equal to the minimum
  positive integer $1 \leq l \leq |V(H)|$ such that the graph $H\oplus
  l$
  has a Hamiltonian cycle. For every $1 \leq l \leq |V(H)|$ the
  graph $H \oplus l$ is equal to the graph $G'(G_1,\dotsc,G_n,I_l)$ where
  $G'$ is the graph obtained from $G$ by adding one vertex $v_{n+1}$ and
  making it adjacent to all vertices of $G$, and the graph $I_l$ is
  the independent set on $l$ vertices. Because
  $\ham(I_l)=|V(I_l)|=l$ we can use
  Lemma~\ref{lem:hamiltonian-cycle-of-product} to decide whether the
  graph $H \oplus l$ has an Hamiltonian cycle in time
  $O(m'2^{n'} + (2m')^{5m'+o(m')}n')$ where $n'=n+1$ and $m'=m+n$. 
  This concludes the proof of the lemma.
\shortversion{\qed}\end{proof}

\section{Conclusion}

We examined some of the algorithmic properties of modular-width, a natural
structural parameter. Our results indicate that this is a notion which may be
worthy of further study independently of its more famous cousin, clique-width,
since it its decreased generality does offer some algorithmic pay-off.

As a direction for further research, it would be interesting to see if more
problems which are hard for clique-width (or even for treewidth) become
tractable for modular-width. Two prime suspects in this category are
\textsc{Edge Dominating Set} and \textsc{Partition into Triangles}.

Beyond that, it would be interesting to investigate if the techniques of this
paper can be further generalized, perhaps eventually leading to
meta-theorem-like results. In particular, our ILP-based solution for
\textsc{Hamiltonicity} may be applicable (with some modifications) to other
problems. One may ask: what properties must a problem possess for us to be able
to give a straightforward DP algorithm that uses ILPs to combine the tables? 

The main property that a problem should satisfy for these ideas to apply is
that the sets of partial solutions arising in the dynamic programming
formulation should be \emph{convex}. Convexity is important here, since we
would like to be able to express the information contained in the DP tables
using linear constraints, in order to use an ILP. Convexity was easy to
establish in the case of \textsc{Partition into Paths} and similar problems,
since the set of feasible partial solutions is the set of integers $k$ such
that there exists a partition of a subgraph into $k$ paths. If one knows the
minimum feasible $k$, all larger integers are also feasible and this is
trivially a convex set. The question then becomes, are there any other natural
problems where convexity can be established (perhaps non-trivially) and used in
this way?

\shortversion{\newpage}

% \bibliographystyle{plain}
% \bibliography{literature}

\begin{thebibliography}{10}

\bibitem{BessyPP10}
St{\'e}phane Bessy, Christophe Paul, and Anthony Perez.
\newblock Polynomial kernels for 3-leaf power graph modification problems.
\newblock {\em Discrete Applied Mathematics}, 158(16):1732--1744, 2010.

\bibitem{BjorklundHK09}
Andreas Bj{\"o}rklund, Thore Husfeldt, and Mikko Koivisto.
\newblock Set partitioning via inclusion-exclusion.
\newblock {\em SIAM J. Comput.}, 39(2):546--563, 2009.

\bibitem{Bodlaender93}
Hans~L. Bodlaender.
\newblock A tourist guide through treewidth.
\newblock {\em Acta Cybernetica}, 11(1--2):1--22, 1993.

\bibitem{Bodlaender00}
Hans~L. Bodlaender.
\newblock The algorithmic theory of treewidth.
\newblock {\em Electronic Notes in Discrete Mathematics}, 5:27--30, 2000.

\bibitem{Bodlaender06}
Hans~L. Bodlaender.
\newblock Treewidth: Characterizations, applications, and computations.
\newblock In Fedor~V. Fomin, editor, {\em Graph-Theoretic Concepts in Computer
  Science, 32nd International Workshop, WG 2006, Bergen, Norway, June 22-24,
  2006, Revised Papers}, volume 4271 of {\em Lecture Notes in Computer
  Science}, pages 1--14. Springer, 2006.

\bibitem{BodlaenderK08}
Hans~L. Bodlaender and Arie M. C.~A. Koster.
\newblock Combinatorial optimization on graphs of bounded treewidth.
\newblock {\em Comput. J.}, 51(3):255--269, 2008.

\bibitem{Bui-XuanTV11}
Binh-Minh Bui-Xuan, Jan~Arne Telle, and Martin Vatshelle.
\newblock Boolean-width of graphs.
\newblock {\em Theor. Comput. Sci.}, 412(39):5187--5204, 2011.

\bibitem{CourcelleMR00}
Bruno Courcelle, Johann~A. Makowsky, and Udi Rotics.
\newblock Linear time solvable optimization problems on graphs of bounded
  clique-width.
\newblock {\em Theory of Computer Systems}, 33(2):125--150, 2000.

\bibitem{diestel00}
Reinhard Diestel.
\newblock {\em Graph Theory}, volume 173 of {\em Graduate Texts in
  Mathematics}.
\newblock Springer Verlag, New York, 2nd edition, 2000.

\bibitem{DomGHNT06}
Michael Dom, Jiong Guo, Falk H{\"u}ffner, Rolf Niedermeier, and Anke Tru{\ss}.
\newblock Fixed-parameter tractability results for feedback set problems in
  tournaments.
\newblock In Tiziana Calamoneri, Irene Finocchi, and Giuseppe~F. Italiano,
  editors, {\em Algorithms and Complexity, 6th Italian Conference, CIAC 2006,
  Rome, Italy, May 29-31, 2006, Proceedings}, volume 3998 of {\em Lecture Notes
  in Computer Science}, pages 320--331. Springer, 2006.

\bibitem{CAP_DOM_SET}
Michael Dom, Daniel Lokshtanov, Saket Saurabh, and Yngve Villanger.
\newblock Capacitated domination and covering: A parameterized perspective.
\newblock In Martin Grohe and Rolf Niedermeier, editors, {\em Parameterized and
  Exact Computation, Third International Workshop, IWPEC 2008, Victoria,
  Canada, May 14-16, 2008. Proceedings}, volume 5018 of {\em Lecture Notes in
  Computer Science}, pages 78--90. Springer, 2008.

\bibitem{DowneyFellows99}
R.~G. Downey and M.~R. Fellows.
\newblock {\em Parameterized Complexity}.
\newblock Monographs in Computer Science. Springer Verlag, New York, 1999.

\bibitem{COL_HARD}
Michael~R. Fellows, Fedor~V. Fomin, Daniel Lokshtanov, Frances~A. Rosamond,
  Saket Saurabh, Stefan Szeider, and Carsten Thomassen.
\newblock On the complexity of some colorful problems parameterized by
  treewidth.
\newblock {\em Inf. Comput.}, 209(2):143--153, 2011.

\bibitem{FellowsLokshtanocMisraRosamondSaurabh08}
Michael~R. Fellows, Daniel Lokshtanov, Neeldhara Misra, Frances~A. Rosamond,
  and Saket Saurabh.
\newblock Graph layout problems parameterized by vertex cover.
\newblock In Seok-Hee Hong, Hiroshi Nagamochi, and Takuro Fukunaga, editors,
  {\em Algorithms and Computation, 19th International Symposium, ISAAC 2008,
  Gold Coast, Australia, December 15-17, 2008. Proceedings}, volume 5369 of
  {\em Lecture Notes in Computer Science}, pages 294--305. Springer, 2008.

\bibitem{FlumGrohe06}
J\"{o}rg Flum and Martin Grohe.
\newblock {\em Parameterized Complexity Theory}, volume XIV of {\em Texts in
  Theoretical Computer Science. An EATCS Series}.
\newblock Springer Verlag, Berlin, 2006.

\bibitem{FominGLS09}
Fedor~V. Fomin, Petr~A. Golovach, Daniel Lokshtanov, and Saket Saurabh.
\newblock Clique-width: on the price of generality.
\newblock In Claire Mathieu, editor, {\em Proceedings of the Twentieth Annual
  ACM-SIAM Symposium on Discrete Algorithms, SODA 2009, New York, NY, USA,
  January 4-6, 2009}, pages 825--834. SIAM, 2009.

\bibitem{FominGLS10b}
Fedor~V. Fomin, Petr~A. Golovach, Daniel Lokshtanov, and Saket Saurabh.
\newblock Algorithmic lower bounds for problems parameterized with
  clique-width.
\newblock In Moses Charikar, editor, {\em Proceedings of the Twenty-First
  Annual ACM-SIAM Symposium on Discrete Algorithms, SODA 2010, Austin, Texas,
  USA, January 17-19, 2010}, pages 493--502. SIAM, 2010.

\bibitem{FominGLS10}
Fedor~V. Fomin, Petr~A. Golovach, Daniel Lokshtanov, and Saket Saurabh.
\newblock Intractability of clique-width parameterizations.
\newblock {\em SIAM J. Comput.}, 39(5):1941--1956, 2010.

\bibitem{Ganian11}
Robert Ganian.
\newblock Twin-cover: Beyond vertex cover in parameterized algorithmics.
\newblock In D{\'a}niel Marx and Peter Rossmanith, editors, {\em Parameterized
  and Exact Computation - 6th International Symposium, IPEC 2011,
  Saarbr{\"u}cken, Germany, September 6-8, 2011. Revised Selected Papers},
  volume 7112 of {\em Lecture Notes in Computer Science}, pages 259--271.
  Springer, 2011.

\bibitem{GanianHNOMR12}
Robert Ganian, Petr Hlinen{\'y}, Jaroslav Nesetril, Jan Obdrz{\'a}lek,
  Patrice~Ossona de~Mendez, and Reshma Ramadurai.
\newblock When trees grow low: Shrubs and fast mso1.
\newblock In Branislav Rovan, Vladimiro Sassone, and Peter Widmayer, editors,
  {\em Mathematical Foundations of Computer Science 2012 - 37th International
  Symposium, MFCS 2012, Bratislava, Slovakia, August 27-31, 2012. Proceedings},
  volume 7464 of {\em Lecture Notes in Computer Science}, pages 419--430.
  Springer, 2012.

\bibitem{habib2010survey}
Michel Habib and Christophe Paul.
\newblock A survey of the algorithmic aspects of modular decomposition.
\newblock {\em Computer Science Review}, 4(1):41--59, 2010.

\bibitem{Oum05}
Sang il~Oum.
\newblock Rank-width and vertex-minors.
\newblock {\em J. Comb. Theory, Ser. B}, 95(1):79--100, 2005.

\bibitem{Lampis12}
Michael Lampis.
\newblock Algorithmic meta-theorems for restrictions of treewidth.
\newblock {\em Algorithmica}, 64(1):19--37, 2012.

\bibitem{McConnellS99}
Ross~M. McConnell and Jeremy Spinrad.
\newblock Modular decomposition and transitive orientation.
\newblock {\em Discrete Mathematics}, 201(1--3):189--241, 1999.

\bibitem{Niedermeier06}
Rolf Niedermeier.
\newblock {\em Invitation to Fixed-Parameter Algorithms}.
\newblock Oxford Lecture Series in Mathematics and its Applications. Oxford
  University Press, Oxford, 2006.

\bibitem{protti2009applying}
F{\'a}bio Protti, Maise~Dantas da~Silva, and Jayme~Luiz Szwarcfiter.
\newblock Applying modular decomposition to parameterized cluster editing
  problems.
\newblock {\em Theory of Computing Systems}, 44(1):91--104, 2009.

\bibitem{TedderCorneilHabibPaul08}
Marc Tedder, Derek~G. Corneil, Michel Habib, and Christophe Paul.
\newblock Simpler linear-time modular decomposition via recursive factorizing
  permutations.
\newblock In Luca Aceto, Ivan Damg{\aa}rd, Leslie~Ann Goldberg, Magn{\'u}s~M.
  Halld{\'o}rsson, Anna Ing{\'o}lfsd{\'o}ttir, and Igor Walukiewicz, editors,
  {\em Automata, Languages and Programming, 35th International Colloquium,
  ICALP 2008, Reykjavik, Iceland, July 7-11, 2008, Proceedings, Part I: Tack A:
  Algorithms, Automata, Complexity, and Games}, volume 5125 of {\em Lecture
  Notes in Computer Science}, pages 634--645. Springer, 2008.

\end{thebibliography}

\end{document}